%% file: main.tex
\documentclass[format=acmsmall, review=false, screen=true]{acmart}
\usepackage{amsmath, amsthm, amssymb}

\usepackage{mathtools}
\usepackage{xcolor}
\usepackage{comment}

\usepackage[noend]{algorithmic}
\usepackage{algorithm}
\usepackage{graphicx}

\usepackage{enumerate}
\usepackage{esvect}

\newtheorem{remark}[theorem]{Remark}

\newtheorem{property}[theorem]{Property}

\newcommand{\eps}{\varepsilon}

\newcommand{\opt}{{\sc OPT}}

\newcommand{\eat}[1]{}

\newcommand{\DRfe}{\textit{DeRand}^F_\eps}
\newcommand{\DR}{\textit{DeRand}}

\newcommand{\floor}[1]{\lfloor #1 \rfloor}
\newcommand{\pdt}{P_{D=t}}
\newcommand{\Alg}{\textit{Alg}}
\newcommand{\Acl}{\textit{Alg}_{\textrm{close}}}
\newcommand{\Afar}{\textit{Alg}_{\textrm{far}}}

\newcommand{\val}{\textrm{Val}}
\newcommand{\I}{\mathcal{I}}

\begin{document} 
\sloppy

\title{Online Allocation with Traffic Spikes: Mixing Adversarial and
  Stochastic Models}

\author{Hossein Esfandiari} 
\affiliation{%
	\institution{University of Maryland}
	\city{College Park}
	\country{USA}}

\author{Nitish Korula}
\affiliation{
	\institution{Google Research}
	\city{NYC}
	\country{USA}}

\author{Vahab Mirrokni}
\affiliation{
	\institution{Google Research}
	\city{NYC}
	\country{USA}}


\begin{abstract}
Motivated by Internet advertising applications, online allocation problems have been studied extensively in various adversarial and stochastic models. While the adversarial arrival models are too pessimistic, many of the stochastic (such as i.i.d or random-order) arrival models do not realistically capture uncertainty in predictions. A significant cause for such uncertainty is the presence of unpredictable traffic spikes, often due to breaking news or similar events. To address this issue, a simultaneous approximation framework has been proposed to develop algorithms that work well both in the adversarial and stochastic models; however, this framework does not enable algorithms that make good use of partially accurate forecasts when making online decisions. In this paper, we propose a robust online stochastic model that captures the nature of traffic spikes in online advertising. In our model, in addition to the stochastic input for which we have good forecasting, an unknown number of impressions arrive that are adversarially chosen. We design algorithms that combine a stochastic algorithm with an online algorithm that adaptively reacts to inaccurate predictions. We provide provable bounds for our new algorithms in this framework. We accompany our positive results with a set of hardness results showing that our algorithms are not far from optimal in this framework. As a byproduct of our results, we also present improved online algorithms for a slight variant of the simultaneous approximation framework.
\end{abstract}


\maketitle

\section{Introduction}
\input{introduction-backup.tex}

\section{Preliminaries and Techniques}
\label{sec:prelims}
\input{prelim.tex}

\section{Robust Online Budgeted Allocation}
\label{sec:robust}
\input{AlmostStochastic.tex}

\section{Approximating Adversarial and Stochastic Budgeted Allocation}
\label{sec:simultaneous}

\input{simultaneous.tex}


\bibliographystyle{ACM-Reference-Format}
\bibliography{robust_budgeted_allocation}



\end{document}

%% file: introduction-backup.tex
In the past decade, online budgeted allocation problems have been studied extensively due to their important applications in Internet Advertising. In such problems, we are given a bipartite graph $G=(X,Y,E)$ with a set of fixed nodes (also known as \emph{agents}, or \emph{advertisers}) $Y$, a set of online nodes (corresponding to \emph{items} or \emph{impressions}) $X$, and a set $E$ of edges between them. Each agent / fixed node $j \in Y$ is associated with a total weighted capacity (or budget) $c_j$; in the context of Internet advertising, each agent corresponds to an advertiser with a fixed budget to spend on showing their ad to users. The items / online nodes $i \in X$ arrive one at a time, along with their incident edges $(i,j)\in E(G)$ and the weights $w_{i,j}$ on these edges. These online nodes correspond to search queries, page-views, or in general, impressions of ads by users. Upon the arrival of an item $i \in X$, the algorithm can assign $i$ to at most one agent $j \in Y$ where $(i,j) \in E(G)$ and the total weight of nodes assigned to $j$ does not exceed $c_j$. The goal is to maximize the total weight of the allocation.

This problem is known as the \emph{Budgeted Allocation} or {\em AdWords} problem, and it has been studied under the assumption that ${\max_{i,j} w_{i,j}\over \min_j c_j} \rightarrow 0$, in~\cite{MSVV,buchbinder-jain-naor,devanur-hayes} (called the \emph{large-degree} assumption). Many variants of this problem such as {\em the display ads} problem \cite{FKMMP09} have been studied, and techniques to solve the budgeted allocation problem have been generalized to solve those problems.

Traditionally, results have been developed for a worst-case arrival model in which the algorithm does not have any prior on the arrival model of online nodes. Under this most basic online model, known as the {\em adversarial model}, the online algorithm does not know anything about the items or $E(G)$ beforehand.  In this model, the seminal result of Karp, Vazirani and Vazirani~\cite{KVV} gives an optimal online $1-{1\over e}$-competitive algorithm to maximize the size of the matching for {\em unweighted graphs} where $w_{ij} = 1$ for each $(i,j)\in E(G)$.  For weighted graphs, Mehta et al.~\cite{MSVV,buchbinder-jain-naor} presented the first $1-{1\over e}$-approximation algorithm to maximize the total weight of the allocation for the AdWords problem.

In practical settings motivated by placement of Internet ads,  the incoming traffic of page-views may be predicted with a reasonable precision using a vast amount of historical data. Motivated by this ability to forecast traffic patterns, various stochastic online arrival models have been introduced. Such models include (i) the i.i.d. stochastic arrival model in which there is a (known or unknown) distribution over the types of items, and each item that arrives is drawn i.i.d. from this distribution~\cite{FMMM09, VeeVS}, or (ii) the random order model~\cite{devanur-hayes, FHKMS10, AWY09}, which makes the weaker assumption that individual items and edge weights can be selected by an adversary, but that they arrive in a random order. Several techniques have been developed to design asymptotically optimal online allocation algorithms for these stochastic arrival models (For example, these algorithms include a set of dual-based algorithms~\cite{devanur-hayes, FHKMS10}, and a set of primal-based algorithms discussed later).

These algorithms for the stochastic models are useful mainly if the incoming traffic of items (i.e. online impressions) can be predicted with high precision. In other words, such algorithms tend to rely heavily on a precise forecast of the online traffic patterns (or if the forecast is not explicitly provided in advance, that the pattern `learnt' by the algorithm is accurate), and hence these algorithms may not react quickly to sudden changes in traffic. In fact, the slow reaction to such traffic spikes imposes a serious limitation in the real-world use of stochastic algorithms for online advertising, and more generally, this is a common issue in applying stochastic optimization techniques to online resource allocation problems (see e.g.,~\cite{Wang06}). To the best of our knowledge, no large Internet advertising systems deploy such stochastic allocation algorithms `as-is' without modifications to deal with situations where the forecasts are inaccurate. Various techniques such as robust or control-based stochastic optimization have been described in the literature~\cite{BN02,BPS04,Wang06,srikant} to deal with this shortcoming, but they do not provide theoretical guarantees when the input is near-adversarial.

One recent theoretical result in this direction is the simultaneous adversarial and stochastic framework~\cite{MOZ11}.  The main question of this recent work is whether there exists an algorithm which simultaneously achieves optimal approximation ratios in both the adversarial and random-order settings. More specifically, does there exist an algorithm achieving a $1-\eps$ approximation for the random-order model, and at the same time achieving a $1-{1\over e}$-approximation for the adversarial model? \citet{MOZ11} showed that the answer to this question is positive for unweighted bipartite graphs, but it is negative for the general budgeted allocation problem. Further, they show that the best $1-{1\over e}$-competitive algorithm for the adversarial model achieves a $0.76$-approximation in the random-order model. Though this shows that the adversarial algorithm has an improved competitive ratio in stochastic settings, it does not use forecast information explicitly, and hence it can be quite far from optimal even when the forecast is perfectly accurate. Moreover, the simultaneous approximation framework is still trying to design an algorithm that is guaranteed to work well in extreme situations (where the input follows the forecast perfectly, or is completely adversarial). What if the forecast is mostly, but not entirely accurate? For instance, suppose traffic to a website largely follows the prediction, but there is a sudden spike due to a breaking news event? Treating this as entirely adversarial input is too pessimistic.

\paragraph{Our Model and Results}{\bf \noindent }
In this paper, we propose a model of online stochastic budgeted allocation with traffic spikes, referred to as {\em Robust Budgeted Allocation}, that goes beyond the worst-case analysis in the adversarial model, and develop algorithms that explicitly use the stochastic information available for arrival pattern. In our model, in addition to the stochastic input for which we have good forecasting, an unknown number of impressions arrive that are adversarially chosen. This model is motivated by the patterns of traffic spikes in online advertising in which part of the incoming traffic of users may be the result of a new event that we did not predict, corresponding to a new traffic pattern. We design an algorithm that adaptively checks if the traffic forecast is accurate, and reacts to flaws in traffic forecasting due to traffic spikes. We measure the accuracy of the forecast in terms of a parameter $\lambda$ which, roughly speaking, captures the fraction of the value of an optimal solution that can be obtained from the stochastic input (as opposed to the adversarially chosen impressions). In general, the competitive ratio of the algorithm will naturally increase with $\lambda$. Furthermore, we accompany our results with a set of hardness results showing that our provable approximation guarantees are not far from the optimal achievable bounds. Hence, compare to the simultaneous approximation framework of~\cite{MOZ11}, our method provides almost-optimal approximation guarantees for all values of $\lambda$, and not only for the extreme cases where either the input is totally adversarial or totally stochastic.

Interestingly, our techniques \emph{also} result in new approaches for the simultaneous approximation framework of~\cite{MOZ11}; though the models are slightly different (i.i.d. vs. random order, as well as the fact that we require a possibly inaccurate forecast of traffic), our algorithm gives improved performance for the weighted case under uncertain input compared to what was achieved in that paper. Section~\ref{sec:prelims} describes the model precisely, allowing us to formally state our results.

Our technique is based on defining a notion of $\eps$-closeness in distributions, and then understanding the behaviour of an online algorithm over sequences that are $\eps$-close to a given distribution. Most notably, we can show how to modify any online stochastic algorithm to work for online adversarial input sequences that are $\eps$-close to a known distribution. This technique is summarized in the next section. We then combine such a modified stochastic algorithm with an adversarial algorithm to guarantee robustness. Converting this idea to provable algorithms for the robust online allocation problem requires applying several combinatorial lemmas and proving invariants that can be converted to a factor-revealing mathematical program, which can then be analyzed numerically and analytically to prove desirable competitive ratios.

\subsection{Other Related Work}
{\bf \noindent Online Stochastic Allocation.}  Two general techniques have been applied to get improved approximation algorithms for online stochastic allocation problems: {\em primal-based} and {\em dual-based} {techniques}.  The dual-based technique is based on solving a dual linear program on a sample instance, and using this dual solution in the online decisions. This method was pioneered by Devanur and Hayes~\cite{devanur-hayes} for the AdWords problem and extended to more general problems~\cite{FHKMS10, AWY09, VeeVS}.  It gives a $1-\eps$-approximations for the random order model if the number of items $m$ is known to the algorithm in advance, and ${\opt\over w_{ij}}\ge O({m\log n \over \eps^3})$, where $n :=|Y|$ is the number of agents.  The primal-based technique is based on solving an offline primal instance, and applying this solution in an online manner.  This method applies the idea of power-of-two choices, and gives improved approximation algorithms for the iid model with known distributions. This technique was initiated by \cite{FMMM09} for the online (unweighted) matching problem and has been improved~\cite{BK10,MOS11,HMZ11,JL13}. All the above algorithms heavily rely on an accurate forecast of the traffic. An alternative technique that has been applied to online stochastic allocation problems is based on optimizing a potential function at each stage of the algorithm~\cite{DJSW11, DevanurSA12}. This technique has been analyzed and proved to produce asymptotically optimal results under the i.i.d. model with unknown distributions. Although this technique does not rely on the accurate predictions as much, it does not combine stochastic and adversarial models, and the analysis techniques used are not applicable to our robust online allocation model.  For unweighted graphs, it has been recently observed that the Karp-Vazirani-Vazirani $1-{1\over e}$-competitive algorithm for the adversarial model also achieves an improved approximation ratio of $0.70$ in the random arrival model~\cite{KMT11,MY11}.  This holds even without the assumption of large degrees. It is known that without this assumption, one cannot achieve an approximation factor better than $0.82$ for this problem (even in the case of i.i.d. draws from a known distribution)~\cite{MOS11}.  All the above results rely on stochastic assumptions and apply only to the random-order or the iid stochastic models.

{\bf \noindent Robust stochastic optimization.} Dealing with traffic spikes and inaccuracy in forecasting the traffic patterns is a central issue in operations research and stochastic optimization.  Methods including robust or control-based stochastic optimization~\cite{BN02,BPS04,Wang06,srikant} have been proposed. These techniques either try to deal with a larger family of stochastic models at once~\cite{BN02,BPS04,Wang06}, try to handle a large class of demand matrices at the same time~\cite{Wang06, AC06, AzarCFKR}, or aim to design asymptotically optimal algorithms that react more adaptively to traffic spikes~\cite{srikant}.  These methods have been applied in particular for traffic engineering~\cite{Wang06} and inter-domain routing~\cite{AC06,AzarCFKR}. Although dealing with similar issues, our approach and results are quite different from the approaches taken in these papers. For example, none of these previous models give theoretical guarantees in the adversarial model while preserving an improved approximation ratio for the stochastic model.  Finally, an interesting related model for combining stochastic and online solutions for the Adwords problem is considered in~\cite{MNS07}, however their approach does not give an improved approximation algorithm for the i.i.d. model.

%% file: prelim.tex
\subsection{Model}
\label{subsec:model}

Let $\I$ denote a set of item `types'; in the Internet advertising applications, these represent queries / ad impressions with different properties that are relevant to advertiser targeting and bidding. A \emph{forecast} $F = (D, f)$ has two components: A distribution $D$ over $\I$, together with a number $f$; this is interpreted as a prediction that $f$ items will arrive, each of which is drawn independently from $D$.

In the \emph{Stochastic Budgeted Allocation problem}, the input known to the algorithm in advance is a forecast $F = (D, f)$, and a set of agents $Y$, with a capacity $c_j$ for agent $j$. A sequence of $f$ items is drawn from the distribution $D$ and sent to the algorithm one at a time; the algorithm must allocate these items as they appear online.  The total weight of items allocated to agent $j$ must not exceed $c_j$, and the objective is to maximize the weight of the allocation. As discussed above, there has been considerable work on near-optimal algorithms for Stochastic Budgeted Allocation~\cite{devanur-hayes, FHKMS10, AWY09, VeeVS, AHL-Prophet}.

In this paper, we define the new \emph{Robust Budgeted Allocation problem}, for which our input model is the following: The adversary can create in advance an arbitrary forecast $F = (D, f)$, and a collection of agents $Y$. Further, at each time step, the adversary can either create a new arbitrary item (together with its incident edges and weights) and send it to the algorithm, or choose to send an item drawn from $D$. After at least $f$ items have been drawn from $D$, the adversary can either send additional (arbitrary) items, or choose to terminate the input. The online algorithm knows in advance only the forecast $F$ and the agents $Y$, and so it knows that it will receive $f$ items corresponding to i.i.d. draws from $D$; it does not know anything about the items created by the adversary, where in the sequence they arrive, or the total number $m$ of items that will arrive. As usual, the competitive ratio of the algorithm is measured by the worst-case ratio (over all inputs) of the value of its allocation to the value of the optimal allocation on the sequence that arrives.

With the preceding description of the model, no algorithm can have a competitive ratio better than $1 - 1/e$, for the simple reason that we could set $f = 0$, allowing the adversary to control the entire input. (Or even for larger $f$, the edge weights for the adversarial items could be considerably larger than the weights for the forecast items in $\I$.) We have not quantified the accuracy of the forecast, or meaningfully limited the power of the adversary. Our goal is to design algorithms with a competitive ratio that improves with the accuracy of the forecast. We quantify this accuracy as follows:

\begin{definition}
For an instance $I$ of the Robust Budgeted Allocation problem with forecast $(D, f)$, let $S(I)$ denote the set of $f$
`stochastic' items drawn from distribution $D$. Let $A(I)$ denote the set of $n-f$ `adversarial' items. When $I$ is
clear from context, we simply use $S$ and $A$ to denote the stochastic and adversarial items respectively. We mildly
abuse notation and, when clear from context, also use $I$ to refer to the sequence of items in an instance.
For a solution $Sol$ to an instance $I$, let $\val_S(Sol)$ denote the value obtained by $Sol$ from allocating the items of $S$ to agents, and $\val_A(Sol)$ denote the value obtained by $Sol$ from allocating the items of $A$.
\end{definition}

\begin{definition}
An instance $I$ of the Robust Budgeted Allocation problem is said to be $\lambda$-stochastic if $\lambda=\max_{Sol\in OPT(I)}\{\frac{E[\val_S(Sol)]}{E[\val_S(Sol) + \val_A(Sol)]}\}$, where $OPT(I)$ is the set of all optimal solutions of $I$.
\end{definition}

Note that when the forecast is completely accurate (there are no adversarial items), the instance is $1$-stochastic, and when $f = 0$ (all items are adversarial), the input is $0$-stochastic. Though $\lambda$ is unknown to the algorithm, our goal is to design algorithms which, when restricted to $\lambda$-stochastic instances, have good competitive ratio (ideally, increasing in $\lambda$).

\subsection{Algorithms for Working with Forecasts}

In this section, we consider how to solve the Stochastic Budgeted Allocation problem. Similar problems have been applied
before (see e.g. \cite{AHL-Prophet}), but we describe a specific approach below that will be a useful tool for the
Robust Budgeted Allocation problem. Further, our argument implies that this approach performs well even for an
\emph{adversarial} input sequence if it is sufficiently `close' to the forecast.

Roughly speaking, given a forecast $F = (D, f)$, if the number of items $f$ is sufficiently large, then we can work with
an `expected instance'. In the expected instance, the set of items is created by assuming each type $t \in \I$ arrives
in proportion to $\pdt$. We then run any (offline) algorithm $\Alg$ on the expected instance; when an item arrives in
the real online instance, we assign it according to the allocation given by $\Alg$ on the expected instance. If the
number of items $f$ is sufficiently large, then only a small error is induced because we assign according to the
expected instance instead of the random realization of the forecast.

\smallskip
We begin by defining the notion of a sequence being \emph{$\eps$-close} to a distribution. Indeed, we show that with high probability a `long' sequence of draws from a distribution is $\eps$-close to that distribution, where $\eps$ is an arbitrary small constant. We next prove that if we ignore inputs which are not $\eps$-close to the input distribution, we lose only $\frac 1 m$ in the competitive ratio.  Finally, we show that we can modify any online stochastic algorithm, or more strongly any offline algorithm for Budgeted Allocation to work for online adversarial input sequences which are guaranteed to be $\eps$-close to a known distribution. Interestingly, this reduction loses only $4\eps$ on the competitive ratio.

\begin{definition} 
Let $S= \langle s_1,s_2,...,s_m \rangle$ be a sequence of items and let $D$ be a distribution over a set of item types $\I$. For a type $t \in \I$, let $P_{D=t}$ be the probability that a draw from $D$ is $t$. We say $S$ is $\eps$-close to distribution $D$, if for any continuous sub-sequence $S_{i,k}=s_i,s_{i+1},\dots,s_k \subset S$ and any type $t$, the number of items of type $t$ in $S_{i,k}$ is within the range $(k-i+1\pm \eps m)P_{D=t}$. If $S$ is not $\eps$-close to distribution $D$ we say it is $\eps$-far from the distribution.
\end{definition}

Consider that $(k-i+1)P_{D=t}$ is the expected number of items of type $t$ in a set of $k-i+1$ draws from $D$. When $k-i+1$ is large, using the Chernoff bound we can show that the number of items of type $t$ is close to the expectation. On the other hand, when $k-i+1$ is small, the term $\eps m$ dominates $k-i+1$, and thus the number of items of type $t$ is in the range $(k-i+1\pm \eps m)P_{D=t}$. Lemma ~\ref{lm:DisGood} formalizes this intuition, showing that with high probability a [long enough] sequence of items drawn from a distribution $D$ is $\eps$-close to $D$. 

\begin{definition}
  Given a distribution $D$, a sequence of $f$ items is said to satisfy the \emph{long input condition} for $D$ if $\frac {f} {\log(f)} \ge \frac {15}{\eps^2 min_{t \in D}(P_{D=t})}$. 
\end{definition}

We use the following version of the Chernoff bound in Lemma ~\ref{lm:DisGood}.

\begin{proposition}
	Let $x_1,x_2,\dots,x_m$ be a set of independent boolean random variables. For $X=\sum_{i=1}^m x_i$ and $\mu = E[X]$ we have: $Pr(|X-\mu| \geq \delta ) \leq 2\exp(\frac {-\delta^2}{3\mu})$
\end{proposition}

\begin{lemma} \label{lm:DisGood}
  Let $S$ be a sequence of $m$ items drawn from a distribution $D$.
  Assuming the long input condition, the probability that $S$ is
  $\eps$-far from $D$ is at most $\frac 1 {m^2}$.
\end{lemma}
\begin{proof}
	$S$ contains $\frac {m(m-1)}{2}$ subsequences. In addition, the long input condition gives us $P_{D=t}\geq \frac
	{15\log(m)}{\eps^2 m}$. This means that there are at most $\frac {\eps^2m} {15 \log(m)}$ different types. Thus, we have
	fewer than $\frac {m^3} 2$ combinations of a type and a subsequence. We next argue that the probability that
	the number of items of a fixed type $t$ in a fixed sub sequence $S_{i,k}$ is out of the range $(k-i+1\pm \eps
	m)P_{D=t}$ is at most $\frac 2 {m^5}$. Applying the union bound, the probability that $S$ is
	$\eps$-far from $D$ is at most $\frac 2 {m^5}\frac {m^3} 2 = \frac 1 {m^2}$, as desired.
	
	For a type $t$ and an index $1\leq \ell \leq m$, let $x_\ell^t$ be a random variable which is $1$ if the $\ell$-th item in
	sequence $S$ is of type $t$ and is $0$ otherwise. The variables $x^t_k$ are independent for a fixed
	type $t$ and different indices. Let $X_{i,k}^t$ be $\sum_{\ell=i}^{k} x_\ell^t$, which is the number of items of type
	$t$ in the sub sequence $S_{i,k}$. The expected value of $X_{i,k}^t$ is $(k-i+1)\pdt$, and by applying
	the Chernoff bound we have:
	\begin{align*}
		Pr(|X_{i,j}^t-(j-i+1)\pdt|\geq \eps m ) &\leq 2 \exp(\frac {-\eps^2m^2}{3(j-i+1)P_{t=D}})\\
		&\leq 2\exp(\frac {-\eps^2m^2}{3n\frac {m\eps^2} {15 \log(m)}})\\
		&=2\exp(- 5\log(m))  =\frac 2 {m^5}.
	\end{align*} 
	This completes the proof of the lemma.
\end{proof}

In the rest of this section, we use the monotonicity and subadditivity properties of Budgeted Allocation, stated in Lemma \ref{Monoton} and Lemma \ref{Subadditive} respectively.

\begin{lemma}[Monotonicity]\label{Monoton}
  Budgeted Allocation is monotone: Fixing the set of agents and their capacities, for any sequence of items $S$ and
  any sequence $T\subseteq S$, we have $Opt(T)\leq Opt(S)$ where $Opt(S)$ and $Opt(T)$ are the values of the optimum
  solutions when the items that arrive are $S$ and $T$ respectively.
\end{lemma}
\begin{proof}
	Any feasible allocation of items of $T$ is a feasible allocation of items of $S$ as well. This immediately means $Opt(T)\leq Opt(S)$ as desired. 
\end{proof}

\begin{lemma}[subadditivity]\label{Subadditive} 
  Budgeted Allocation is subadditive: Fixing the set of agents and their capacities, for any sequence of items $S$
  and any sequence of items $T$, we have $Opt(S\cup T)\leq Opt(S)+Opt(T)$ where $Opt(X)$ indicates the size of the
  optimum solution when the sequence of items that arrive is X.
\end{lemma}
\begin{proof}
Fix an optimum solution $Opt(S\cup T)$. The allocation of items of $S$ in $Opt(S\cup T)$ is a feasible allocation for $S$. Similarly the allocation of items of $T$ in $Opt(S\cup T)$ is a feasible allocation for $T$. Therefore we have $Opt(S\cup T)\leq Opt(S)+Opt(T)$.
\end{proof}

Lemma \ref{lm:DisGood} says that w.h.p, a sequence of items drawn from a distribution $D$ is $\eps$-close to $D$. That
is, inputs which are $\eps$-far from the input distribution are rare, but this does not immediately imply that the total
value of such rare $\eps$-far inputs is small as well. Lemma \ref{lm:DisGood2} says that we may ignore all inputs which
are $\eps$-far from the input distribution and only lose a small fraction in the competitive ratio.

\begin{lemma} \label{lm:DisGood2}
  Let $S$ be a sequence of $m$ items drawn from a distribution $D$, satisfying the long input condition. Let $\Alg$ be
  an $\alpha$-competitive algorithm for Stochastic Budgeted Allocation with forecast $(D, m)$. Let $\Acl$
  be an algorithm that has the same outcome as $\Alg$ when the input is $\eps$-close to $D$ and $0$ otherwise.
  $\Acl$ is $(\alpha-\frac {1}{ m} )$-competitive.
\end{lemma}

\begin{proof}
	Let $\Afar$ be an algorithm that has the same outcome as $\Alg$ when the input is $\eps$-far from $D$ and $0$
	otherwise. We slightly abuse notation and use $\Alg, \Acl, \Afar$ to refer both to the algorithms and the expected
	values of their outcomes. By definition we have $\Alg=\Acl+\Afar$. Let $Opt$ denote the expected value of an optimal
	solution on the sequence drawn from the distribution. We bound the expected outcome of $\Afar$ to compare
	the competitive ratio of $\Alg$ and $\Acl$.
	
	Let $S_{k}$ be an input that contains $k$ items of each type. The monotonicity of Budgeted Allocation implies that for
	any sequence $S$ of size $m$, $Opt(S_m)$ is greater than $Opt(S)$. Moreover, by subadditivity, $Opt(S_m)$ is at most
	$\frac m 2$ times $Opt(S_2)$. Together, these imply that, for any sequence $S$ of size $m$, we have $Opt(S)\leq \frac
	m 2 Opt(S_2)$.
	 On the other hand, the number of items of type $t$ in any sequence $S$ of size $m$ which is $\eps$-close to $D$ is at least
	 \begin{align*}
	 	(1 - \eps)fP_{D=t} \geq
	 	 (1 - \eps)\frac {15}{\eps^2 min_{t \in D}(P_{D=t})}P_{D=t}\geq 
	 	 (1 - \eps)\frac {15}{\eps^2}\geq 2.	 	 	 	 
	 \end{align*}
	Thus by monotonicity we have $\frac 1 {Pr(S\text{ is } \eps \text{ close to }D)} Opt_{close} \geq
	Opt(S_2)$. Therefore, for \emph{any} sequence $S$ of size $m$ we have
	
	\begin{align*}
		Opt(S)\leq \frac m 2 Opt(S_2)\leq \frac m 2 \frac 1 {(1-\frac 1 {m^2})} Opt_{close}  \leq m Opt
	\end{align*}
	This together with Lemma \ref{lm:DisGood} gives us
	\begin{align*}
		\Acl = \Alg - \Afar \geq \Alg - m Opt\frac 1 {m^2} \geq \alpha Opt - \frac 1 m Opt = (\alpha-\frac 1 m) Opt 
	\end{align*}
	as desired.
\end{proof}

Note that any algorithm $\Alg$ for Budgeted Allocation has a random outcome when the items are drawn from a forecast,
simply due to the randomness in the sequence of items. We now define a derandomization of such algorithms: Given an
algorithm $\Alg$, a forecast $F = (D, f)$ and a constant $\eps$, algorithm $\DRfe(\Alg)$ is defined as follows:

Let $S'$ be a sequence of $(1-\eps)f$ impressions, with $(1-\eps)f \pdt$ impressions of type $t$, for each type
$t$. Run algorithm $\Alg$ on sequence $S'$. Let $\Alg(S',t,i)$ be the agent to which $\Alg$ assigns the $i$th impression of
type $t$ in $S'$. Note that any sequence of $f$ items which is $\eps$-close to $D$ contains at least $(1-\eps)f 
\pdt $ impressions of each type $t$. We can now describe how $\DRfe(\Alg)$ allocates items of a sequence $S$. For each
type $t$, Algorithm $\DRfe(\Alg)$ allocates the first $(1-\eps) f \cdot \pdt$ impressions of type $t$ in $S$ in the same
manner as $\Alg$ allocated $S'$. That is, we assign the $i$th impression of type $t$ in $S$ to $\Alg(S',t,i)$. After the
first $(1 - \eps) f \pdt$ impressions of type $t$ have been allocated, we do not allocate any more items of this type.
Finally, if at any time during the algorithm, we observe that the input sequence (so far) is not $\eps$-close to
distribution $D$, the algorithm stops and returns false. Otherwise, it returns true.

When it is clear from the context, we drop $F$ and $\eps$ from the notation of $\DRfe(\Alg)$.

\begin{remark} \label{re:DeRand}
  Note that for any forecast $F = (D, f)$ and constant $\eps$, the outcome of $\DRfe(\Alg)$ on any item sequence of
  length $f$ that is $\eps$-close to $D$ is a function purely of $D$ and $\eps$, but not the actual impressions in the
  sequence.
\end{remark}

\begin{theorem}\label{thm:CloseIsGood}
  Let $F = (D, f)$ be a distribution, and let $A$ be an adversarial input with length $f$ such that $A$ satisfies the long
  input condition and $A$ is $\eps$-close to $D$. Let $\Alg$ be an $\alpha$-competitive algorithm for Stochastic
  Budgeted Allocation. Though $A$ is not explicitly drawn from $D$, $\DRfe(\Alg)$ is $\alpha-2\eps$ competitive on
  sequence $A$. 
\end{theorem}
\begin{proof}
	Since $A$ is $\eps$-close to $D$, $\DR(\Alg)$ will not return false on $S$, and thus, the allocation of $\DR(\Alg)$ on
	$S$ is identical to that of $\Alg$ on $S'$. 
	Let $S''$ be a sequence of $(1+\eps)d$ impressions, $(1+\eps)f \pdt $ from each type $t$.
	Using the subadditivity of Budgeted Allocation, we have $Opt(S') \ge \frac {1 - \eps} {1+\eps}Opt(S'') \geq (1-2\eps)Opt(S'')$.
	
	On the other hand, using the monotonicity of Budgeted Allocation, $Opt(A)\leq Opt(S'')$. Together, these inequalities
	imply that $Opt(S')\geq (1-2\eps) Opt(A)$, which means that $\DRfe(\Alg)$ is $(\alpha-2\eps)$-competitive.
\end{proof}

Consider an $\alpha$-competitive online algorithm (or $\alpha$-approximate offline algorithm), $\Alg$, for Budgeted
Allocation with stochastic input. By Theorem \ref{thm:CloseIsGood}, for inputs which are $\eps$-close to $D$,
$\DRfe(\Alg)$ has a competitive ratio that is only $2\eps$ worse than that of $\Alg$. Moreover, Lemma \ref{lm:DisGood2}
says that if we ignore all inputs which are $\eps$-far from the input distribution, we lose only $\frac 1 m$ on the
competitive ratio. Together with the assumption that $\eps \geq \frac 1 m$, we obtain the following corollary.

\begin{corollary}\label{cr:R}
  Let $\Alg$ be an $\alpha$-competitive algorithm for Stochastic Budgeted Allocation (or
  $\alpha$-approximate offline algorithm for Budgeted Allocation). Assuming the long
  input condition, for any small constant $\eps\geq \frac 1 m$ and any forecast $F$,
  $\DRfe(\Alg)$ is an $(\alpha-3\eps)$-competitive algorithm for Stochastic Budgeted Allocation.
\end{corollary}

%% file: AlmostStochastic.tex
In this section, we consider the Robust Budgeted Allocation problem. As described above, the algorithm knows in advance the set of agents $Y$, together with a capacity $c_j$ for each agent $j \in Y$. Further, it knows the forecast $F = (D, f)$, but not how many additional items will be sent by the adversary, nor how they are intermixed with the $f$ items drawn from $D$. 

Recall that a $\lambda$-stochastic instance is one where the `stochastic items' (those drawn from $D$) provide $\lambda$-fraction of the value of an optimal solution. (Also, note that $\lambda$ is not known to the algorithm.)  As $\lambda$ increases (corresponding to an increase in the accuracy of our forecast, or a smaller traffic spike), we expect the performance of our algorithms to improve. However, we wish our algorithms to be robust, obtaining good performance compared to an optimal offline solution even when $\lambda$ is close to $0$ (corresponding to a very large traffic spike, when the typical `forecast' traffic is only a small fraction of the total).

First, in Section~\ref{subsec:unweighted}, we consider the unweighted Robust Budgeted Allocation problem, in which $w_{ij}$ is the same for all $(i, j) \in E$. As desired, we obtain an algorithm with competitive ratio tending to $1$ as $\lambda$ tends to $1$, and $1 - 1/e$ when $\lambda$ tends to $0$. Then, in Section~\ref{subsec:weighted}, we consider the general weighted Robust Budgeted Allocation problem, and finally, in Section~\ref{subsec:hardness}, we give upper bounds on the competitive ratio of any algorithm. All our competitive ratios are parametrized by $\lambda$, and they are summarized in Figure~\ref{fig:all_curves}. 

\begin{figure}
  \begin{center}
    \includegraphics[width=12cm]{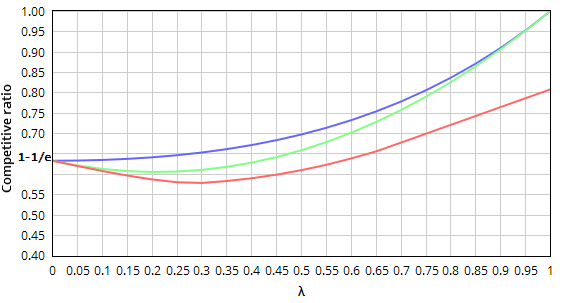}
  \end{center}
  \caption{Summary of results, parametrized by $\lambda$. The blue curve denotes the upper bound on the competitive ratio of any algorithm, the green curve is our algorithm for the unweighted case, and the red curve is our algorithm for the weighted case.}
  \label{fig:all_curves}
\end{figure}

For simplicity, throughout this section, we assume that the capacity of all the agents are the same (normalized to $1$). This assumption can be removed by dividing agents with large capacities into multiple dummy agents. Moreover, as usual we assume $\max_{i,j} w_{i,j}\rightarrow 0$  (a.k.a., large degree assumption).

\subsection{Unweighted Robust Budgeted Allocation}
\label{subsec:unweighted}

Our general approach will be to \emph{simulate} the following idea: Suppose we first receive the $f$ stochastic items drawn from $D$. We could allocate these items optimally (assuming they are the entire input). If this uses up most of the budgets of the agents, we automatically have a good solution, even if we do not allocate any of the adversarial items. If some fraction of the budgets remain unused, use this remaining capacity to allocate any adversarial items that arrive. If there is a way to allocate these adversarial items using the remaining capacities, we obtain $1-1/e$ fraction of this value. 

Unfortunately, in the real Robust Budgeted Allocation model, we do not know which items are stochastic and which are adversarial, so we cannot perform this clean separation perfectly. Still, we can approximate this separation as follows: Let $A$ be an algorithm for Stochastic Budgeted Allocation. We have two algorithms running simultaneously. The first is a slight variant of $\DRfe(A)$, and the second is Balance~\cite{kp-balance}. More precisely, we send each item to $\DRfe(A)$; recall that this algorithm first checks if the input seen so far is $\eps$-close to $D$. If it is, it allocates this item according to $A$; otherwise, it returns false. Now, instead of returning false, we assume that this item must have been sent by the adversary. As such, we erase this item from the history of $\DRfe(A)$, and try to allocate it using Balance. That is, we `guess' that all items that are matched by $\DRfe(A)$ are stochastic items and all other items are adversarial.

Note that $\DRfe(A)$ does not match more than $(1-\eps)f$ items, $(1-\eps)f\pdt$ from each type $t$. From Lemmas~\ref{lm:DisGood} and \ref{lm:DisGood2}, we know that w.h.p., the sequence of stochastic items is $\eps$-close to $D$; by ignoring the case when it is $\eps$-far from $D$, we lose at most $\frac{1}{m}$ in the competitive ratio. Given the long input assumption, $1/m < \eps$, and hence by losing $\eps$ in the competitive ratio, we assume that the stochastic items are always $\eps$-close to $D$. Since the sequence of stochastic items is $\eps$-close to $D$, we have at least $(1-\eps)f \pdt$ items of type $t$. Therefore, the items that $\DRfe(A)$ leaves unallocated are a superset of all the adversarial items. (More precisely, there may be an adversarial item of type $t$ that we guess is stochastic, but there must be a corresponding stochastic item of type $t$ that we treat as adverarial instead.)

We now complete the description of our combined algorithm: Given an allocation, let $x_j$ denote the fraction of agent $j$'s budget used in this allocation. Let $\Alg_S$ be a $(1-\eps)$-competitive algorithm for Stochastic Budgeted Allocation with the further property that it minimizes $\sum_{j=1}^n x_j^2$. (In other words, $Alg_S$ reproduces an optimal offline solution such that each item is allocated to an eligible agent with the lowest $x_j$.) We run $\DRfe(\Alg_S)$ as the first of our two algorithms. Recall from Remark~\ref{re:DeRand} we know the exact allocation of $\DRfe(Alg_S)$, which is independent of the random draws. If $x_j$ denotes the fraction of agent $j$'s capacity used by this allocation, for items unallocated by $\DRfe(Alg_S)$, we run the Balance algorithm on the instance of Budgeted Allocation with the adversarial items, and in which the capacity of agent $j$ is $1-x_j$. 

How does this combined algorithm perform? We use $S$ to denote the stochastic items, and $A$ the adversarial items. From Corollary~\ref{cr:R} we know that $\DRfe(\Alg_S)$ is $1-O(\eps)$ competitive against $Opt(S)$. We will prove in Lemma \ref{lm:goodAdv} that the optimum solution on the items of $A$ using the remaining capacities $1 - x_j$ is at least $(1-\lambda-O(\eps))(Opt(A\cup S)-Opt(S))$. Since Balance is a $(1 - 1/e)$-competitive algorithm, the value we derive from Balance is at least $(1-\frac 1 e) (1-\lambda-O(\eps)) (Opt(A\cup S)-Opt(S))$. Therefore, the competitive ratio of the combined algorithm is at least:
\begin{align*}
&\frac {(1-O(\eps))Opt(S)+(1-\lambda-O(\eps))(1-\frac 1 e)(Opt(A\cup S)-Opt(S))}{Opt(S\cup A)} \\
& = \frac {\lambda Opt(S\cup A)+(1-\lambda)^2(1-\frac 1 e)Opt(S\cup A)-O(\eps)Opt(S\cup A)}{Opt(S\cup A)} \\
&= \lambda+(1-\lambda)^2(1-\frac 1 e)-O(\eps)
\end{align*}
where the first equality uses the definition of $\lambda$-stochastic to replace $Opt(S)$ with $\lambda Opt(S \cup A)$. This proves the following theorem.

\begin{theorem}
Assuming the long input condition, there exists a $\lambda+(1-\lambda)^2(1-\frac 1 e)-O(\eps)$ competitive algorithm for Robust Budgeted Allocation when the input is $\lambda$-stochastic. This algorithm does not depend on $\lambda$.
\end{theorem}

In order to prove the key Lemma~\ref{lm:goodAdv}, we first write Mathematical Program~\ref{fig:MPunW}, and show in Lemma~\ref{lm:MPisAns} that this lower bounds the ratio of the optimum allocation of the adversarial items on the remaining capacities to the overall optimum. Next, in Lemma~\ref{lm:goodMP}, we show that the solution of this mathematical program is at least $1-\lambda$. Lemma~\ref{lm:goodAdv} is an immediate result of combining these two lemmas.

\begin{figure}
\begin{center}
\fbox{\parbox{.97\columnwidth}{
\parbox{.94\columnwidth}{
\begin{alignat*}{6}
\text{Minimize:}\;\;\;&\mathrlap{\frac{\sum_{j=1}^{n}{(y_j-z_j)}}{\sum_{j=1}^{n}{y_j}}}&
\\
\text{Subject to:}\;\;\;
&z_j\leq \max(0,x_j+y_j-1)&\;\;\; \forall 1\leq j\leq n \\
&\sum_{k=j+1}^{n}{(x_k+y_k)}\leq n-j &\;\;\; \forall 1\leq j \leq n-1\text{ s.t. } x_j<x_{j+1}\\
&x_j\leq x_{j+1} &\;\;\; \forall 1\leq j \leq n-1\\
&\lambda(\sum_{j =1}^{n}{x_j}+\sum_{j =1}^{n}{y_j})\geq \sum_{j =1}^{n}{x_j}\\
& 0 \leq x_j,y_j,z_j\leq 1 &\;\;\; \forall 1\leq j \leq n
\end{alignat*}
}}}
\end{center}
\caption{Mathematical Program 2 to bound the optimum allocation of adversarial items on the remaining capacities in uniform weight case.}
\label{fig:MPunW}
\end{figure}

We now show that there exists an optimum allocation of all items such that the contribution of stochastic items in this allocation is exactly $Opt(S)$. We will need this for Lemma~\ref{lm:MPisAns}.

\begin{lemma}\label{lm:goodSinOpt}
  For any disjoint sets of items $A$ and $S$, there exist an optimum allocation of $S\cup A$ such that the contribution of items of $S$ in this allocation is exactly $Opt(S)$.
\end{lemma}

\begin{proof} 
	We prove this lemma by contradiction. Assume that in all optimum allocations of $S \cup A$, the contribution of items of $S$ is less than $Opt(S)$. Consider an optimum allocation $Opt'(S\cup A)$ of $S \cup A$ such that the items of $S$ have the maximum contribution to the allocation. Denote the allocation of items of $S$ in $Opt'(S\cup A)$ by $A'(S)$. By definition, $A'(S)$ is not an optimum allocation of $S$. Thus, there exists an augmenting path in $A'(S)$ which can increase the number of assigned items to one of the agents by exactly one, and keep the number of assigned items to all other agents the same. This change increases the number of assigned items in $A'(S)$ by one, and may decrease the number of assigned items from $A$ by at most one. Therefore, it is an optimal solution of $S\cup A$ in which items from $S$ have more contribution, and hence gives a contradiction. 
\end{proof}

\begin{lemma}\label{lm:MPisAns}
  The optimum allocation of the adversarial items on the remaining capacities is at least $Opt(S\cup A)-Opt(S)$ times the solution of Mathematical Program~\ref{fig:MPunW}.
\end{lemma}
\begin{proof}
Let $x_j$ be the fraction of the capacity of the agent $j$ which is filled by $\Alg_S$. Without loss of generality, assume that $x_j$s are in increasing order. That is, for all $1\leq j\leq n-1$, we have $x_j\leq x_{j+1}$, which is the third constraint in Mathematical Program \ref{fig:MPunW}.

Consider the optimum solution of $S\cup A$ from Lemma \ref{lm:goodSinOpt}. Let $y_j$ be the fraction of the capacity of the agent $j$ that is filled by the adversarial items in this solution. One can see that the fourth constraint says that the contribution of the stochastic items is at most $\lambda$ fraction of the total value. (From Lemma~\ref{lm:goodSinOpt}, we could have equality in the fourth constraint, but for simplicity of analysis, we maintain the inequality.)

Note that we want to compare the optimum solution of the adversarial items on the remaining capacities $1 - x_j$ with the total value  $\sum_{j=1}^{n}{y_j}$ of the adversarial items in the optimum solution of $S \cup A$. For some agent $j$, if we have $y_j+x_j\leq 1$, we can assign the same adversarial items to the agent $j$ as in the optimum allocation. On the other hand, if $y_j+x_j \geq 1$, we can only use $1-x_j$ fraction of the capacity of agent $j$ for the adversarial items. Thus, we can always assign $y_j-max(0,x_j+y_j-1)$ fraction of the capacity of agent $i$ to the adversarial items. This quantity is denoted by $z_j$ in the first constraint.

By assigning adversarial items as in the optimum solution of $S \cup A$ using the remaining capacities, we can obtain at least $\sum_{j=1}^{n}{y_j-z_j}$. The objective function of Mathematical Program \ref{fig:MPunW} compares this to what the optimum solution of $S \cup A$ can get from adversarial items. 

Finally, fix an index $j$ such that $x_j<x_{j+1}$ and look at all agents with index greater than $j$. All stochastic items that we match to agents with index greater than $j$ have no edge to the agents with index less than or equal to $j$. (By definition of $\Alg_S$, which assigns items to the eligible agent with lowest value of $x_j$.) Thus, in any optimum solution of $S$, they cover at least $\sum_{k=j+1}^n x_k$ of the agents with index greater than $j$. Thus, this optimum solution of $S \cup A$ covers $\sum_{k=j+1}^n (x_k+y_k)$ of the agents with index greater than $j$. Consider that we have $n-j$ such agents, we get $\sum_{k=j+1}^n (x_k+y_k)\leq n-j$. This is the second constraint in Mathematical program \ref{fig:MPunW}.
\end{proof}

\begin{lemma}\label{lm:goodMP}
For any small number $\delta$, the solution of the mathematical program \ref{fig:MPunW} is at least $1-\lambda-O(\delta)$.
\end{lemma}

\begin{proof}
	First, we show that given any solution to the Mathematical Program~\ref{fig:MPunW}, we can construct a related solution with a very restricted structure that only increases the objective function by $O(\delta)$. We will lower bound the objective value of this new solution, which will imply a lower bound on the original competitive ratio. 
	
	Our restricted solution will satisfy Properties \ref{pr:MP1}, \ref{pr:MP2}, \ref{pr:MP2.5}, \ref{pr:MP3} ,\ref{pr:MP4}, and \ref{pr:MP5}, which we state below. Finally, we use these properties to bound the value of the objective function by $1-\lambda$. This means that the solution of MP \ref{fig:MPunW} is bounded by $1-\lambda-O(\delta)$.
	
	\begin{property}\label{pr:MP1}
		Without loss of generality we can assume that for any positive number $\delta$, we have, $\frac{1}{\sum_{j=1}^n y_j}\leq \delta$.
	\end{property}
	
	To prove this property given an instance of the Mathematical program \ref{fig:MPunW}, we provide another instance that satisfy this property and has a solution equal to the optimum solution of the initial instance.
	Consider a solution $Sol^n=<(x_1,\dots,x_n),(y_1,\dots,y_n),(z_1,\dots,z_n)>$, to the MP \ref{fig:MPunW} with $n$ agents. For any integer $C$, one can construct an equivalent feasible solution, $Sol^{Cn}=<(x'_1,\dots,x'_{Cn}),(y'_1,\dots,y'_{Cn}),(z'_1,\dots,z'_{Cn})>$ to an extended Mathematical Program equivalent to \ref{fig:MPunW} with $Cn$ agents as follow. For all $1\leq j\leq n$ and all $1 \leq k \leq C$, set 
	\begin{align*}
		x'_{C(j-1)+k}=x_j,\;\; y'_{C(j-1)+k}=y_j,\;\;
		z'_{C(j-1)+k}=z_j.
	\end{align*}
	One can easily check that if $Sol^n$ is feasible, $Sol^{Cn}$ is feasible as well. The value of the objective functions in both $Sol^n$ and $Sol^{Cn}$ are the same. Thus, if we bound the value of the objective function in $Sol^{Cn}$, this bound holds for $Sol^n$ as well. Moreover, since we have $C$ copies of each variable here, we have $\sum_{j=1}^{Cn}y'_j = C\sum_{j=1}^{n}y_j$. Hence, if we set $C=\frac{1}{\delta\sum_{j=1}^{n}y_j }$, it satisfies Property~\ref{pr:MP1}. In the rest of the proof, we assume that the solution that we are considering has Property \ref{pr:MP1}, and for simplicity we refer to it as a solution with $n$ agents and variables $x_j$s, $y_j$s and $z_j$s.
	
	Note that Property \ref{pr:MP1} means that if we decrease the value of a constant number of $z_j$s to $0$, we only increase the objective value by $O(\delta)$. We later use this to provide Property \ref{pr:MP2}.
	
	\begin{definition}
		We say an agent $j$ is \emph{harmful} if we have $z_j>0$. Otherwise, we say the agent is a sink.
	\end{definition}
	
	If an agent $j$ is harmful, $z_j$ is positive and thus it contributes to reducing the objective value below $1$. 
	
	\begin{property}\label{pr:MP2}
		All $y_j$s are either $0$ or $1$.
	\end{property}
	Let $j$ be the largest index such that the agent $j$ is harmful and $y_j$ is neither $1$ nor $0$. If this is not the only harmful agent with this property, there is some other index $k$ such that agent $k$ is harmful and $y_k$ is neither $1$ nor $0$. We decrease the value of $y_j$ and $z_j$ by some $\eps$ and increase $y_j$ and $z_k$ by $\eps$. This change has no effect on the objective function, or the third and fourth constraints. It increases both sides of the first constraint by $\eps$ for agent $k$ and decreases both sides by $\eps$ for agent $j$, keeping them both feasible.
	One can verify that this change may only decrease the left hand side of the second constraint for indices between $k$ and $j$. 
	
	We can repeat this procedure until there is just one harmful agent $j$ with a fractional $y_j$. We decrease $y_j$ and $z_j$ of this agent to $0$, and by property \ref{pr:MP1}, only increase the objective value by $O(\delta)$. Note that harmful agents with $y_j = 0$ now become sink agents.
	
	We can use the same procedure to make the $y_j$s of all sink agents either $0$ or $1$. 
	
	\begin{property}\label{pr:MP2.5}
		If agent $j$ is harmful, $z_j=x_j$. 
	\end{property}
	For each harmful agent $j$, the value of $z_j$ can be as much as $x_j$ (since $y_j$ for a harmful agent is now $1$). Increasing $z_j$s of harmful agents up to their $x_j$s may only decrease the objective function, and keeps all the constraints feasible. 
	
	From now on, for each harmful agent $j$, we keep $x_j$ and $z_j$ the same, and if we change one of these two, we change the other one as well. This implicit change of $z_j$ keeps the constraints feasible. However, each time we decrease $x_j$ of a harmful agent, we need to guarantee that the objective function does not increase. This guarantee is trivial when we are increasing the value of some other $x_k$ to match.
	
	\begin{property}\label{pr:MP3}
		If an agent $j>1$ is a sink, we have $x_{j-1}=x_j$.
	\end{property}
	Let $j$ and $h$ be two consecutive harmful agents with $j<h$. For all agents $k$ in the range $[j,h)$ we replace the value of $x_k$ with the average value of $x$s in this range i.e. we set $x_k = \frac{\sum_{k=j}^{h-1}x_k}{h-j}$. This keeps the sum of all $x_k$s the same, and keeps all of the constraints feasible. Moreover, since $x_j$ has the minimum value in the range $[j,h)$, and agent $j$ is the only harmful agent in this range, we are only increasing the $x$ valuess of harmful agents. This only decreases the objective function.  
	
	\medskip
	Let $k$ be the number of different $x_j$ values that we have, and let $\alpha_h$ be the $h$-th smallest value among distinct $x_j$s. For simplicity, we assume we always have some $x_j=0$. Let $\beta_h$ be the number of $x_j$s which are at most $\alpha_h$, and set $\beta_0=0$. Thus, $\beta_{h-1}\leq j<\beta_{h}$ means that $x_j=\alpha_h$.

	\begin{property}\label{pr:MP4}
		For any $1\leq h\leq k$, all harmful agents in the range $(\beta_{-h1},\beta{h}]$, are the first agents in the range.
	\end{property}
	
	Let $i$ be a sink and let $i+1$ be a harmful agent, such that $x_i=x_{i+1}$. We can replace the value of the variables for these two agents and make $i$ harmful, instead of $i+1$. It is easy to see that this keeps the mathematical program feasible. By repeating this, for each range $(\beta_{j-1},\beta{j}]$, we can move all harmful agents to be the first agents in the range.
	
	\begin{property}\label{pr:MP5}
		All $x_j$s are either $0$ or the same fixed number $x^*$.
	\end{property}
	
	We iteratively decrease the number of distinct $x_j$s as follow.
	
	Let $\xi$ be some small number that we set later. It is easy to see that one of the following two actions decreases the objective function, since the effect of one on the objective function is the negative of the effect of the other one. The action that decreases the objective function is the one that we do.
	\begin{itemize}
		\item For all $\beta_1< j \leq \beta_2$ decrease $x_j$ by $\frac{\xi}{\beta_2-\beta_1}$ and for all $\beta_2< \ell \leq \beta_3$ increase $x_\ell$ by $\frac{\xi}{\beta_3-\beta_2}$
		
		\item For all $\beta_1< j \leq \beta_2$ increase $x_j$ by $\frac{\xi}{\beta_2-\beta_1}$ and for all $\beta_2< \ell \leq \beta_3$ decrease $x_\ell$ by $\frac{\xi}{\beta_3-\beta_2}$
	\end{itemize}
	
	We can set $\xi$, such that the mathematical program remains feasible and one of the following situations happen.
	\begin{itemize}
		\item $x_j$s in the range $(\beta_1,\beta_2]$ are all $0$.
		\item $x_j$s in the range $(\beta_2,\beta_3]$ are all $x_{\beta_3+1}$.
		\item The second constraint of the MP is tight for $j=\beta_2$.
	\end{itemize}
	In the first two cases, the number of distinct $x_i$s is reduced by one and we are done in these cases. 
	
	Therefore, we assume that we are in the third case. Now suppose that $k>3$ (that is, there are more than $3$ distinct values of $x_j$); later, we consider the case that $k=3$. Let $\xi'$ be some small number that we set later. For all $\beta_2< j \leq \beta_{k-1}$ we increase $x_j$ by $\frac{\xi'}{\beta_{k-1}-\beta_{2}}$ and for all $\beta_{k-1}< \ell \leq \beta_k$ we decrease $x_\ell$ by $\frac{\xi'}{\beta_{k}-\beta_{k-1}}$. Below, we show that by doing so, the objective function does not increase. We can set $\xi'$ to make the largest and second largest value of $x_j$s become equal. This decreases the number of distinct values of $x_i$s.
	
	The second constraint in MP \ref{fig:MPunW} for $j=\beta_{k-1}$ says that the fraction of harmful agents in the range $(\beta_{k-1},\beta_k]$ is at most $1-x_n$. This constraint is tight for $j=\beta_2$. Thus, the fraction of harmful agents in the range $(\beta_2,\beta_k]$ is at least $1-x_n$, which means the fraction of harmful agents in the range $(\beta_2,\beta_{k-1}]$ is at least $1-x_n$. Thus, taking some fraction of $x_\ell$s in the range $(\beta_{k-1},\beta_k]$ and distributing it equally over the range $(\beta_2,\beta_{k-1}]$ increases the sum of $z_j$s and thus decreases the objective function.
	
	In the case that $k=3$, we move the harmful agents in the range $(\beta_1,\beta_2]$ to some higher indices, to make the second constraint tight. Then, we can use the above technique and decrease the value of the only two non zero distinct $x_j$s to become equal.   
	
	Now, given the property \ref{pr:MP5} we can easily bound the objective function of the mathematical program as follow. If $x^*\leq \lambda$, all $z_j$s are at most $\lambda$, while $y_j$s are all $1$. This bounds the objective function by $1-\lambda$. On the other hand, if $x^*>\lambda$, the second constraint in MP \ref{fig:MPunW} says that the number of harmful agents is at most $1-x^*$ fraction of the value from the stochastic items. This means that sum of $z_j$s is at most $1-x^*$ fraction of value of the stochastic items. Thus, it is at most $(1-x^*)\frac{\lambda}{1-\lambda}\leq \lambda$ and this bounds the objective function by $1-\lambda$.
\end{proof}

The following lemma is an immediate result of combining Lemma \ref{lm:MPisAns} and Lemma \ref{lm:goodMP}.

\begin{lemma}\label{lm:goodAdv}
  The optimum allocation of adversarial items on the remaining capacities is at least $(1-\lambda-O(\delta))(Opt(S\cup A)-Opt(S))$.
\end{lemma}








\subsection{Weighted Robust Budgeted Allocation}
\label{subsec:weighted}
\input{WeightedAlgorithm.tex}

\subsection{Hardness of Robust Budgeted Allocation}
\label{subsec:hardness}

\begin{theorem}\label{thm:hardness}
  No online algorithm for Robust Budgeted Allocation (even in the unweighted case) has competitive ratio better than $1 - \frac{1-\lambda}{e^{ 1-\lambda}}$ for $\lambda$-stochastic inputs.
\end{theorem}

\begin{proof}
	Consider the following instance: We have $n$ agents each with capacity one and $n$ items. Initially, all of the agents are \emph{unmarked}. Each item is connected to all agents which are unmarked upon its arrival. The first $\lambda n$ items are connected to all agents. After the arrival of the first $\lambda n$ items, we \emph{mark} the $\lambda n$ agents with minimum expected load. Subsequently, after assigning each item we \emph{mark} the agent with the minimum expected load among all unmarked agents.
	
	Note that all of the first $\lambda n$ items are adjacent to all agents, and there is no uncertainty about them. Thus, one can consider the first $\lambda n$ items as being part of the forecast $F = (D, \lambda n)$, where the distribution $D$ has a single type (items adjacent to all agents). The last $(1-\lambda)n$ vertices as adversarially chosen items. Let $v_j$ denote the $j$th agent that we mark. An optimum algorithm that knows the whole graph in advance can match the $j$th item to the $j$th agent and obtain a matching of size $n$.
	
	Now consider an arbitrary online algorithm. Let $\ell^i_j$ be the expected load of agent $v_j$ after assigning the $i$th item and let $\ell_j = \ell^n_j$ be its expected load at the end of the algorithm. One can see that $\ell_j=\ell_j^{\max(\lambda n,j)}$. The first $\lambda n$ items put a load of $\lambda n$ on the agents and agents $v_1$ to $v_{\lambda n}$ are the $\lambda n$ agents with lowest expected loads. Thus, we have $$ \sum_{j=1}^{\lambda n} \ell_j \leq \lambda^2 n.$$
	
	In addition, for each $\lambda n < j \leq n$, the first $j$ items put a load of at most $j$ on the agents, and the expected load on the agents $v_1$ to $v_{j-1}$ is $\sum_{k=1}^{j-1} \ell_k$. Therefore, we can bound the expected load of the $j$th agent as follows: $$ \ell_j \leq \frac {j-\sum_{k=1}^{j-1} \ell_k} {n-j+1}.$$
	
	Further, the expected load of an agent cannot exceed $1$, and the optimum solution is a matching in which each agent has a load of $1$. Therefore, the competitive ratio of any algorithm for Robust Budgeted Allocation on $\lambda$-stochastic inputs is bounded by Linear program \ref{fig:LPunW}.

	\begin{figure}
		\begin{center}
			\fbox{\parbox{.97\columnwidth}{
					\parbox{.94\columnwidth}{
						\begin{alignat*}{6}
							\text{Maximize:}\;\;\;\;\;&\mathrlap{\frac {\sum_{j=1}^{n} \ell_j}{n}}&
							\\
							\text{Subject to:}\;\;\;\;\;
							&\sum_{j=1}^{\lambda n} \ell_j \leq \lambda^2 n \\
							&\ell_j \leq \frac {j-\sum_{k=1}^{j-1} \ell_k} {n-j+1}&\;\;\; \forall \; {\lambda n+1\leq j\leq n}\\
							&0 \leq \ell_j \leq 1 &\;\;\;\;\; \forall \; {1\leq j \leq n}
						\end{alignat*}
					}}}
				\end{center}
				\caption{Linear program to bound the competitive ratio of the best online algorithm.}
				\label{fig:LPunW}
			\end{figure}
			
			First note that w.l.o.g., we can assume that in any optimal solution to LP~\ref{fig:LPunW}, each $\ell_j$ for $1 \le j \le \lambda n$ has the same value. Consider an optimal solution in which as many constraints as possible are tight. It is easy to see that if the constraint for $\ell_j$ is not tight, we can raise $\ell_j$ by $\eps$ and decrease $\ell_{k}$ by at most $\eps$ (to maintain feasibility) for the first $k > j$ such that $\ell_k > 0$. By repeating this procedure iteratively, we make all of the constraints tight. This solution is correspond to the algorithm that matches items to their eligible agents uniformly. Hence, it allocates all $\lambda n$ stochastic items and at most $\xi n$ adversarial items, for some $\xi$ such that 
			$\sum_{i=1}^{\xi n}  \frac {1}{(1-\lambda)n-(i-1)}\geq 1-\lambda$. Notice that we have
			\begin{align*}
				\sum_{i=1}^{\xi n}  \frac {1}{(1-\lambda)n-(i-1)} \geq \int_{(1-\lambda-\xi)n}^{(1-\lambda)n} \frac 1 x dx = \log((1-\lambda)n) - \log((1-\lambda-\xi)n) = \log (\frac {1-\lambda}{1-(\lambda+\xi)}).
			\end{align*}
			By setting $\log (\frac {1-\lambda}{1-(\lambda+\xi)}) = 1-\lambda$ and rearranging the parameters we have  $\lambda+\xi = 1 - \frac{1 - \lambda}{e^{1-\lambda}}$.
\end{proof}

%

%% file: WeightedAlgorithm.tex
We can now describe our algorithm for general weighted case of Budgeted Allocation. As in the unweighted case, we combine the allocations of two algorithms: One that runs on the (stochastic) items that we guess are drawn from the forecast, and one on the items that we guess are constructed by the adversary. 

For the stochastic items, we start with a base algorithm for Stochastic Budgeted Allocation that is \emph{not necessarily optimal}. Instead of maximizing the weight of the allocation from the stochastic items, we start with an algorithm $\Alg_{pot}$ that \emph{maximizes the potential} of the allocation, as defined below. 

\begin{definition}\label{def:Pot}
  Let $X=(x_1,x_2,\dots,x_n)$ be a vector of numbers, such that for all $1\leq j\leq n$, we have $0 \leq x_j\leq 1$. We define the potential of $x_j$, $Pot(x_j)$ to be $x_j-e^{(x_j-1)}$. We define the potential of the vector $X$, $Pot(X)$ to be $\sum_{j=1}^n Pot(x_j)$.
\end{definition}

Let $x_j$ denote the fraction of capacity $c_j$ used by the potential-maximizing allocation of $\Alg_{pot}$. Similarly to the unweighted case, when items arrive, we send them to $\DR(\Alg_{pot})$; for those items that are unmatched, by $\DR(\Alg_{pot})$, we send them to the Balance algorithm using the remaining capacities $1 - x_j$. Exactly the same argument that we provide for the unweighted case works here to show that by losing $O(\eps)$ on the competitive ratio, we can assume that we match all stochastic items using $\DR(\Alg_{pot})$ and all adversarial items using the Balance algorithm. We use $\Alg$ to denote this combined algorithm.
In order to analyze our algorithm $\Alg$,  we need to define another potential function based on $Pot(X)$. 

\begin{definition}
  Let $X=(x_1,x_2,\dots,x_n)$ and $Y=(y_1,y_2,\dots,y_n)$ be two vectors of numbers between $0$ and $1$. For each $1\leq j\leq n$, we define $Pot_X(y_j)$ as follow. If $x_j\leq y_j$, we have $Pot_X(y_j)=Pot(y_j)$. Otherwise, we have $Pot_X(y_j)=Pot(x_j)+(y_j-x_j)Pot'(x_j)$, where $Pot'(.)$ is the first derivative of $Pot(.)$. Thus, for $x_j<y_j$ we have:
  \begin{align*} Pot_X(y_j)=x_j-e^{x_j-1}+(y_j-x_j)(1-e^{x_j-1}) \end{align*}
 We define $Pot_X(Y)$ to be $\sum_{j=1}^n Pot_X(y_j)$.
\end{definition}

Note that the second derivative of $Pot(x_j)$ is $-e^{x_j-1}$ which is always negative. Thus, $Pot(x_j)$ is a concave function. Notice that in the range $[0,x_j]$, $Pot_X(y_j)$ is equal to $Pot(y_j)$ and in the range $(x_j,1]$, it is the tangent line to $Pot(x_j)$ at $x_j$. Hence, $Pot_X(y_j)$ is a concave function as well.

Consider that, in the range $(x_j,1]$, the function $Pot_X(y_j)$ is the degree $2$ Taylor series of $Pot(y_j)$ at point $x_j$. In addition, the second derivative of $Pot(y_j)$ in the range $[0,1]$ is lower-bounded by $-1$, yielding the following lemma.   

\begin{lemma} \label{lm:PotxIsClose}
  For any constant $\eps$ and any vector of positive numbers $X$ and any $y_j$ such that $0\leq y_j \leq x_j+\eps \leq 1$ we have $|Pot_X(y_j)-Pot(y_j)| \le \eps^2$.
\end{lemma}

We now have the tools to analyze $\Alg$ via means of a mathematical program: First, in Lemma \ref{lm:WMPisAns}, we show that the competitive ratio of our algorithm is lower bounded by the solution of the Mathematical program \ref{fig:WMP}. Next, in Lemma \ref{lm:goodWMP}, we lower bound the solution of the mathematical program; this lower bound is shown in Figure~\ref{fig:weighted}. Together, these lemmas allow us to prove the following theorem:

\begin{figure}
\begin{center}
\includegraphics[width=11cm]{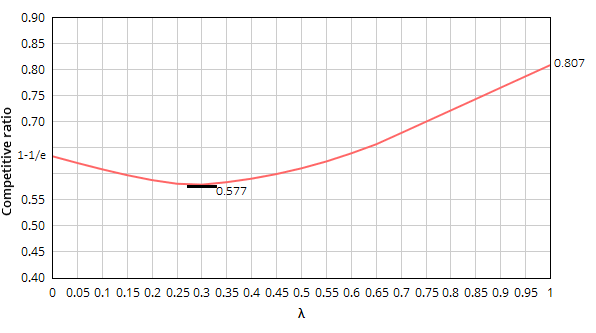}
\end{center}
\caption{Competitive ratio of $Alg$ parametrized by $\lambda$.}
\label{fig:weighted}
\end{figure}

\begin{theorem}
  There exists an algorithm for $\lambda$-stochastic weighted budgeted allocation with competitive ratio presented in Figure~\ref{fig:weighted}.
\end{theorem}

\begin{figure}
\begin{center}
\fbox{\parbox{.97\columnwidth}{
\parbox{.94\columnwidth}{
\begin{alignat}{6}
\text{Minimize:}\;\;\;&\mathrlap{\frac{\sum_{j=1}^{n}{x_j}+(1-\frac 1 e)\sum_{j=1}^{n}{(y_j-z_j)}}{\sum_{j=1}^{n}{y_j}+\sum_{j=1}^{n}{t_j}}}&\label{WMP:obj}
\\
\text{Subject to:}\;\;\;
&z_j\leq \max(0,x_j+y_j-1)&\;\;\; \forall 1\leq j\leq n\label{WMP:z} \\
& Pot_X(T)\leq Pot(X) &\label{WMP:Pot}\\
& t_j+y_j\leq 1 &\;\;\; \forall 1\leq j \leq n\label{WMP:ty1}\\
&\lambda(\sum_{j=1}^{n}{t_j}+\sum_{j=1}^{n}{y_j})= \sum_{j = 1}^{n}{t_j}\label{WMP:lambda}\\
& 0 \leq t_j,x_j,y_j,z_j\leq 1 &\;\;\; \forall 1\leq j \leq n
\end{alignat}
}}}
\end{center}
\caption{Mathematical Program 4 to bound the competitive ratio of $Alg$.}
\label{fig:WMP}
\end{figure}




\begin{lemma}\label{lm:WMPisAns}
  The competitive ratio of $Alg$ is bounded by the solution of the Mathematical program \ref{fig:WMP}.
\end{lemma}
\begin{proof}
  For each agent $j$, let $x_j$ be the expected fraction of $c_j$ used by $\Alg_{pot}$ on the stochastic items. Consider the optimum assignment $Opt_{\lambda}$ that maximizes $\frac{E[\textit{Stochastic}(Opt_{\lambda})]}{E[Opt_{\lambda}]}$. For each agent $j$, let $t_j$ and $y_j$ respectively denote the expected fraction of $c_j$ used by stochastic items and adversarial items in $Opt_{\lambda}$. Therefore, the expected fraction of $c_j$ used in $Opt_{\lambda}$ is $t_j+y_j$, which is upper-bounded by $1$. This gives us the Inequality \ref{WMP:ty1}.

By definition of $\lambda$, the contribution of the stochastic items to $Opt_{\lambda}$ is $\lambda$ fraction of the total value; equality \ref{WMP:lambda} captures this fact.

For any agent $j$, if we have $y_j\leq 1-x_j$, an offline algorithm can allocate all the same adversarial items to agent $j$ as $Opt_{\lambda}$ did by using the remaining $1 - x_j$ fraction of $c_j$. On the other hand, if we have $y_j>1-x_j$, the offline algorithm can use the entire remaining capacitity $1-x_j$ of agent $j$. Thus, an offline algorithm that assigns all the adversarial items using the $1-x_j$ fraction of capacities that remain (after allocating the stochastic items according to $\Alg_{pot}$)  only loses at most $\max(0,x_j+y_j-1)$ from agent $j$ when compared to the allocation of adverarial items by $Opt_{\lambda}$.  We denote $\max(0,x_j+y_j-1)$ by $z_j$, as shown in Inequality \ref{WMP:z}. 

One can see that the numerator of the objective function lower-bounds the expected outcome of $Alg$ and the denominator is the expected outcome of $Opt_{\lambda}$. Thus, the objective function is a lower bound on the competitive ratio of $Alg$.,

It remains only to verify inequality \ref{WMP:Pot}. By definition of $\Alg_{pot}$ as an algorithm maximizing potential, we know that $Pot(T)\leq Pot(X)$. However, this does not imply the inequality directly. We will show that if the inequality does not hold, we will be able to construct a new allocation with a larger potential, contradicting the definition of $X$.  Suppose by way of contradiction that there exists some $T$, and some positive number $\delta$ such that $Pot_X(T) - Pot(X)\geq \delta$.  For some arbitrary small constant $\eta$, consider the vector $\eta T + (1-\eta)X$, which allocates $\eta$ fraction of each stochastic item according to $T$ and $1-\eta$ fraction of each stochastic item according to $X$. By concavity of $Pot_X(.)$ we have
 \[ Pot_X(\eta T+(1-\eta)X)\geq \eta Pot_X(T)+(1-\eta)Pot_X(X)=\eta Pot_X(T)+(1-\eta)Pot(X).  \]

Together with the assumption that $Pot_X(T) - Pot(X)\geq \delta$, this gives $Pot_X(\eta T+(1-\eta)X)\geq Pot(x)+\eta\delta$. On the other hand, Lemma~\ref{lm:PotxIsClose} implies that $Pot_X(\eta T+(1-\eta)X)\leq Pot(x)+n \eta^2$. By setting $\eta$ to be smaller than $\frac {\delta}{n}$, we obtain a contradiction. Thus, for any solution $T$ we have $Pot_X(T)\leq Pot(X)$ which is the inequality \ref{WMP:Pot}.  \end{proof}




\begin{lemma}\label{lm:goodWMP}
  The solution of Mathematical Program \ref{fig:WMP} is lower-bounded by the curve in Figure~\ref{fig:weighted}.
\end{lemma}
\begin{proof}
Again, as in the unweighted case, given any solution to Mathematical program \ref{fig:WMP}, we find a new solution with a restricted structure that only increases the objective function by $O(\delta)$. Eventually, we lower bound the objective value of this new solution; this provides a lower bound on the original competitive ratio. 

Our restricted solution will satisfy properties \ref{pr:W1}, \ref{pr:W2}, \ref{pr:W3} and \ref{pr:W4} below.  Similar to the unweighted case, we can replace each variable with with $O(1/\delta)$ copies with the same value in order to satisfy Property~\ref{pr:W1}, which implies that if we change a constant number of the $y_j$ and / or $t_j$ variables, it changes the objective function by at most $O(\delta)$.

\begin{property}\label{pr:W1}
  For any positive number $\delta$, we can create an equivalent solution such that $\frac{1}{\sum_{j=1}^n y_j}\leq \delta$
  and $\frac{1}{\sum_{j=1}^n t_j} \le \delta$.
\end{property}

Again as in the unweighted section, we call an agent $j$ with positive $z_j$ a \emph{harmful} agent. We say that an agent $j$ is a \emph{source} if both $z_j$ and $t_j$ are zero. If an agent is neither harmful nor a source, we call it a \emph{sink}.

\begin{property}\label{pr:W2}
  All $y_j$s and $t_j$s are either $0$ or $1$.
\end{property}
To satisfy this property, we add some new extra agents with $x_.=0$, $y_.=0$, $z_.=0$, $t_.=0$ initially. (Note that this requires slightly modifying the mathematical program to add constraints for the new agents, and including them in the summations. However, it is easy to see that this does not change the objective or feasibility of any constraint.) 

Now, for any original agent $j$ with positive $y_j$ and such that $j$ is either a sink or a source, we decrease $y_j$ to 0 and increase $y_k$ by the same amount for some newly added agent $k$. It is easy to see that this does not affect the sum of the $y$ variables (and hence the objective function or the feasibility of any constraint), and that this can be done so that only a single agent has a fractional value of $y$. 

Let $\Gamma$ be an integer such that the sum of $t_j$s for all sinks is in the range $[\Gamma,\Gamma+1)$. (Note that sources have $t_j = 0$.) We redistribute the $t_j$ values to be $1$ on the $\Gamma$ sources and sinks with the maximum $x_j$s, leaving at most one sink with fractional $t_j$. Using the concavity of $Pot_X(T)$, this change only decreases $Pot_X(T)$. Since all sources and sinks now have $y_j = 0$, we still have $t_j + y_j \le 1$ for each agent $j$. (Note that some agents may change status from source to sink or vice versa during this process.)

Now, $y_j$ and $t_j$ are either $0$ or $1$ for all except the harmful agents. (Further, $y_j$ is $1$ only on newly added agents, and $t_j$ is $1$ only on sink agents.)

For the harmful agents, let $Y_H$ and $T_H$ denote the sum of their $y_j$ and $t_j$ values respectively. From constraint~\eqref{WMP:ty1}, we know that $Y_H + T_H$ is less than the total number of harmful agents. We redistribute the $y_j$ and $t_j$ values as follows: Set $y_j = 1$ on the $\floor{Y_H}$ harmful agents with the smallest values of $x_j$, and $0$ otherwise; similarly, set $t_j = 1$ on the $\floor{T_H}$ harmful agents with the largest values of $x_j$, and 0 otherwise. In this process, we lose the fractional $y_j$ and $t_j$ values of at most one agent, possibly increasing the competitive ratio by  $O(\delta)$. Further, note that this process leaves some of these agents harmful (those with $y_j = 1$), while the others (all but $\floor{Y_H}$ of them) are no longer harmful; $\floor{T_H}$ of them now become sinks. There are two further effects of this move: First, as argued above for the sinks, by setting $t_j = 1$ on the agents with highest $x_j$, $Pot_X(T)$ only decreases. Second, the redistribution of $y_j$ values affects the $z_j$ values. To satisfy constraint~\eqref{WMP:z}, we must decrease $z_j$ to $0$ on each agent which is no longer harmful. However, we can increase the $z_j$ values on the the $\floor{Y_H}$ agents which are still harmful, and it is easy to see that this only decreases the objective function (because the increase in these $z_j$ values exceeds the decrease on the agents which are no longer harmful).

Therefore, we satisfy Property~\ref{pr:W2} while increasing the objective by at most $O(\delta)$. 

\medskip
For each harmful vertex $i$, Inequality~\eqref{WMP:z} says that $z_j\leq x_j$ (since $y_j = 1$). Increasing $z_j$ to $x_j$ only decreases the objective function. Thus, in the rest of the proof, for any harmful vertex $j$ we assume $z_j=x_j$. Changing $x_j$ of a harmful vertex $j$ implicitly affects $z_j$ as well.

\begin{property}\label{pr:W3}
  All harmful agents have the same $x_j$, denoted by $\alpha_1$. All sink agents have the same $x_j$, denoted by $\alpha _2$.
\end{property}

We replace the $x_j$ of all harmful agents with their average. By concavity of the potential function, this does not decrease $Pot(X)$. (Note that all harmful agents $j$ have $t_j = 0$.) It also does not change $\sum_{j=1}^n x_j$ and $\sum_{j=1}^n z_j$, and thus, keeps all the constraints feasible.

We also replace the $x_j$ of sinks with their average.  One can rewrite Inequality~\ref{WMP:Pot} as $\sum_{j=1}^n (Pot_X(t_j)-Pot(x_j))\leq 0$. The left term for sink $j$ is $Pot(x_j)+(1-x_j)(1-e^{x_j-1})-Pot(x_j)=(1-x_j)(1-e^{x_j-1})$. Thus, for each sink, the left term is a convex function. This means that this change does not increase the left hand side of $\sum_{j=1}^n (Pot_X(t_j)-Pot(x_j))\leq 0$. This change also does not affect $\sum_{j=1}^n x_j$ and $\sum_{j=1}^n z_j$, and hence keeps all the constraints feasible.

\medskip
We use $\beta_1$ to denote the number of agents $j$ with $y_j = 1$; this includes harmful agents and newly added agents. We use $\beta_2$ and $\beta_3$ to denote the number of sinks and sources respectively. We use $\gamma$ to denote $\sum_j x_j$ for sources.

\begin{property}\label{pr:W4}
$\sum_{i\text{ is source}} Pot(x_i)$ is bounded by $(1-\frac 1 e)\gamma-\frac 1 e \beta_3$
\end{property}
This follows from the fact that $Pot(x)$ is always less than $(1-\frac 1 e)x-\frac 1 e$.

Now, using Properties~\ref{pr:W2}, \ref{pr:W3} and \ref{pr:W4}, one can show that the objective function is lower bounded by
\begin{align}\label{eq:obj}
  \frac{(\alpha_1\beta_1+\alpha_2\beta_2+\gamma)+(1-\frac{1}{e})(1-\alpha_1)\beta_1}
  {\beta_1+\beta_2}.
\end{align}
(This is a lower bound because in the numerator, we have $\sum_(y_j - z_j)$, which is $1$ for the newly added agents included in $\beta_1$, and $(1 - \alpha_1)$ for the harmful agents.)

We can write Equality~\eqref{WMP:lambda} as
\begin{align}\label{eq:lambda}
  \lambda (\beta_1+\beta_2)= \beta_2,
\end{align}
and write Inequality~\eqref{WMP:Pot} as
\begin{align*}
  &\beta_1(0-e^{-1})+\beta_2(\alpha_2-e^{\alpha_2-1}+(1-\alpha_2)(1-e^{\alpha_2-1}))+\beta_3(0-e^{-1})
  \leq\\
  &\qquad \beta_1(\alpha_1-e^{\alpha_1-1})+\beta_2(\alpha_2-e^{\alpha_2-1})+(1-\frac 1 e)\gamma+\beta_3(0-e^{-1}).
\end{align*}

One can drop the terms $\beta_2(\alpha_2-e^{\alpha_2-1})$ and $\beta_3(0-e^{-1})$ from both sides of the inequality, and rearrange as follows:

\[
\frac e {e-1}(\beta_1(-\frac 1 e-\alpha_1+e^{\alpha_1-1})+\beta_2(1-\alpha_2)(1-e^{\alpha_2-1})) \leq \gamma.
\]

We replace $\gamma$ in the objective function (Equation \ref{eq:obj} with the max of zero and left hand side of the inequality. Then, we replace $\beta_1$ with $\frac {1-\lambda}{\lambda}\beta_2$ using Equality~\eqref{eq:lambda}. Now we can cancel out $\beta_2$ and simplify the objective function as follows.
\begin{align*}
&\frac{\alpha_1}{e}(1-\lambda)+(1-\lambda)(1-\frac{1}{e})+\lambda\alpha_2+\\
&\max(0,\frac{e}{e-1}((1-\lambda)(-\frac{1}{e}-\alpha_1+e^{\alpha_1-1})+\lambda(1-\alpha_2)(1-e^{\alpha_2-1})))
\end{align*}

Unfortunately, it is hard to solve this analytically to obtain a closed-form experssion for the competitive ratio as a function of $\lambda$. When $\lambda \ge 0.6882$, we can lower bound the competitive ratio by the line $0.3714 + 0.4362 \lambda$. However, this is not necessarily tight. Numerically solving, the curve lower bounding the competitive ratio as a function of $\lambda$ can be seen in Figure~\ref{fig:weighted}
\end{proof}

%% file: simultaneous.tex
In this section, we study a class of algorithms for Budgeted
Allocation problem that provide good approximation ratios in both
stochastic and adversarial settings. We say an algorithm is
$(\alpha,\beta)$-competitive if it is $\alpha$-competitive in the
adversarial setting and $\beta$-competitive in the stochastic setting.

The best algorithms for the stochastic setting have a competitive
ratio of $1-\eps$ when the input is stochastic. However, these
algorithms may have a competitive ratio close to zero when the input
is adversarial. On the other hand, the Balance algorithm has the best
possible competitive ratio of $1-\frac 1 e$ when the input is
adversarial, but, under some mild assumption, is $0.76$-competitive
for stochastic inputs~\cite{MOZ11}. Indeed, this result holds when, 1)
the capacities are large i.e. $\max_{i,j}\frac{w_{i,j}}{c_j}\leq
\eps^3$, and 2) the optimum solution is large
i.e. $\eps^{-6}\sum_j \max_{i:opt(i)=j}w_{i,j}\leq Opt$, where
$Opt(i)$ is the agent the optimum matches to node $i$, and $\eps$ is a vanishingly small number. In this
section, we refer to $\max_{i,j}\frac{w_{i,j}}{c_j}\leq \eps^3$ as the
large capacities condition, and refer to $\max\left(
  \eps^{-6}\sum_j \max_{i:opt(i)=j}w_{i,j} , \eps^{1/3} m\cdot
  \max_{i,j}w_{i,j} \right)\leq Opt$ as the large optimum condition.

Of course, one can design algorithms with `intermediate' performance:
Simply use the stochastic algorithm with probability $p$ and Balance
with probability $1-p$; this yields an algorithm with competitive
ratios of $((1-p)(1-\frac 1 e)$ and $(p+(1-p)\times0.76)-O(\eps))$ in
the adversarial and stochastic settings respectively. These
competitive ratios for different values of $p$ lie on the straight
line with endpoints $(1-\frac 1 e, 0.76)$ and $(0,1-\eps)$. In this
section, we use the tools that we developed in
Section~\ref{sec:prelims} to obtain competitive ratios those are
`above' this straight line.

We devise a mixed algorithm for this setting as follows: Divide the
capacity of each agent / offline vertex into two parts. The first part
has $1-p$ fraction of the capacity, and the second part has the
remaining $p$ fraction of the capacity. Given a forecast $F = (D, f)$,
at the beginning, run the Balance algorithm on $1-p$ fraction of the
forecast input (that is, on the first $(1-p) \cdot f$ items, using the
first $(1-p)$ fraction of the capacities. Then, for any
$1-\eps$-competitive stochastic algorithm $\Alg$, run $\Alg$ on the
rest of the forecast input (that is, the next $p f$ items), using the
remaining $p$ fraction of the capacities.  If at any time during the
algorithm, we observe that the input sequence is $\eps$-far from $D$,
or the length of the input is more than $f$, we can detect that
(w.h.p.) we are in the adversarial setting. We therefore flush out the
items we assigned to the second part (the $p$ fraction) of the
capacities, if any, and return to the Balance algorithm. However, we
now run two independent copies of Balance: The first using $1-p$
fraction of the capacities (with some items previously allocated by
Balance), and the second using the remaining $p$ fraction of the
capacities (beginning with an empty allocation of items to this part
of the capacities). For each item arriving online, with probability
$1-p$ we send it to the Balance algorithm on the first part of the
capacities, and with probability $p$ we send it to the Balance
algorithm on the second part of the capacities.

In the next two lemmas, we bound the competitive ratios of our mixed algorithm when the input is stochastic and
adversarial respectively. 

\begin{lemma}\label{lm:stoch}
  Assuming the long input, the large capacities and the large optimum
  conditions, if the input is stochastic, for any constant $p\in [0,1]$ the mixed algorithm has a
  competitive ratio of $(1-p)*\beta+ p - O(\eps)$, where $\beta$ is
  $0.76$.
\end{lemma}
\begin{proof}
  From Theorem \ref{thm:CloseIsGood}, we know that by ignoring the inputs which are $\eps$-far from the distribution, we
  lose only $O(\eps)$ fraction in the competitive ratio. 
  We use $1-p$ fraction of the input sequence on the first $1 - p$ fraction of the capacities. Therefore, the Balance
  algorithm gets $(1-p)*\beta$ fraction of the optimal solution from this part of the input, where $\beta =
  0.76$ is the competitive ratio for Balance on stochastic inputs~\cite{MOZ11}. On the other hand, we use $p$ fraction
  of the input on the remaining $p$ fraction of the capacities. Since we use a $1-\eps$ competitive stochastic algorithm, we get $p-O(\eps)$
  fraction of the optimal solution from this part of the input sequence.
\end{proof}

\begin{lemma} \label{lm:adver}
  Assuming the long input, the large capacities and the large optimum
  conditions, if the input is adversarial, for any constant $p\in [0,1]$ the mixed algorithm has a
  competitive ratio of $\frac {0.76
    (1-p)}{1+0.202(1-p)}-O(\eps^{1/3})$.
\end{lemma}

\begin{proof}
	Let $S$ be the maximal prefix of the input before observing that the input is not $\eps$-close to the distribution and let $A$ be the rest of the
	input. (Note that either $S$ or $A$ may be empty.) If $S$ contains at least $\epsilon^{2/3} m$ nodes, since $S$ is $\frac{\eps}{\epsilon^{2/3}}=\epsilon^{1/3}$-close to the distribution and satisfies the long input condition, we obtain at least
	$\beta (1-p-O(\epsilon^{1/3})) Opt(S)$ from this fraction of the input (Note that if $S$ is smaller than $(1-p) \cdot f$, for example,
	we would obtain a larger fraction of $Opt(S)$.). On the other hand, if  $S$ contains less than $\epsilon^{2/3} m$ nodes, using large optimum condition, we have $Opt(S)\leq \epsilon^{2/3} m\cdot max_{i,j}w_{i,j} \leq \epsilon^{1/3} Opt(A\cup S)$. Thus, the competitive ratio is lower bounded by
	\[
	\frac {\beta(1-p)Opt(S)}{Opt(S\cup A)}-O(\epsilon^{1/3})\geq \frac {\beta(1-p)Opt(S)}{Opt(S)+Opt(A)}-O(\epsilon^{1/3})\label{eq:lb1}.
	\]
	
	On the other hand, if we ignore $\epsilon m$ nodes, the Balance algorithm on the first $1-p$ fraction of the capacities always runs on at least
	$1-p$ fraction of the items, we always get at least $(1-\frac 1 e )(1-p-O(\epsilon^{2/3})) Opt(S\cup A)$ from this fraction of the
	capacities. Further, once the input stops being $\eps$-close to the distribution (that is, for the entire sequence
	$A$), we use the remaining $p$ fraction of the capacities on $p$ fraction of $A$. Therefore, we always get at least
	$(1-\frac 1 e )p Opt(A)$ from this $p$ fraction of the capacities. Thus, the competitive ratio is lower bounded by
	\begin{align*}
		&\frac {(1-\frac 1 e )(1-p-O(\epsilon^{2/3}))Opt(S\cup A)+(1-\frac 1 e )p Opt(A)}{Opt(S\cup A)}=\\
		&\left(1-\frac 1 e \right)(1-p)+\frac {(1-\frac 1
			e )p Opt(A)}{Opt(S\cup A)}-O(\epsilon^{2/3})\geq\\
		&\left(1-\frac 1 e \right)(1-p)+\frac {(1-\frac 1 e )p Opt(A)}{Opt(S)+Opt(A)}-O(\epsilon^{2/3})
	\end{align*}
	where the inequality follows from the subadditivity of Budgeted Allocation.  
	
	The first of these lower bounds is decreasing in $Opt(A)$, while the second lower bound is increasing in
	$Opt(A)$. Thus, the worst case happen when these two bounds are equal. By setting the two equal and $\beta=0.76$, one
	can achieve the desired bound.
\end{proof}

The following theorem is the immediate result of combining Lemmas~\ref{lm:stoch} and \ref{lm:adver}.

\begin{theorem}
  Assuming the long input, the large capacities and the large optimum
  conditions, for any constant $p \in [0,1]$ the mixed algorithm is
  $(1-p)*0.76+ p-O(\eps)$ competitive for stochastic input and
  $\frac {0.76 (1-p)}{1+0.202(1-p)}-O(\eps^{1/3})$ competitive for
  adversarial input.
\end{theorem}